\documentclass[preprint,12pt]{article}
\usepackage{hyperref}
\usepackage{graphicx}
\usepackage{multirow}
\usepackage{amsmath,amssymb,amsfonts}
\usepackage{amsthm}
\usepackage{authblk}
\usepackage{mathrsfs}
\usepackage[title]{appendix}
\usepackage{xcolor}
\usepackage{textcomp}
\usepackage{manyfoot}
\usepackage{booktabs}
\usepackage{algorithmic}
\usepackage[ruled]{algorithm2e}
\usepackage{listings}

\newcommand{\NN}{\mathbb N}
\newcommand{\ZZ}{\mathbb Z}
\newcommand{\QQ}{\mathbb Q}
\newcommand{\CC}{\mathbb C}

\newtheorem{theorem}{Theorem}

\begin{document}
\title{Improved Asymptotic Runtime of Gotzmann's Decomposition}
\author[1]{Joseph Donato\thanks{joseph.donato.gr@dartmouth.edu}}
\affil{Dartmouth College}
\date{}
\maketitle

\begin{abstract}
In the study of Hilbert schemes, the integer partition or Gotzmann's decomposition of a Hilbert polynomial helps researchers identify some geometric and combinatorial properties of the scheme in question.
To aid researchers in extracting such information from a Hilbert polynomial, we describe an efficient algorithm which can identify if $p(x)\in\QQ[x]$ is a Hilbert polynomial and if so, recover the integer partition $\lambda$ associated with it.
\end{abstract}
{\bf Keywords:} Hilbert Scheme, Hilbert Polynomials, Discrete Calculus, Polynomial Interpolation\\
{\bf Disclosure Statement:} The author reports there are no competing interests to declare.

\section{Introduction}
The roots of this work can be traced back to the following theorem proved by Macaulay.
\begin{theorem}
\cite{Mac} Let $R=\CC[x_{0},...,x_{v}]$ and $p(x)\in\QQ[x]$, there exists ideals in $R$ with Hilbert polynomial $p(x)$ iff $p(x)$ can be written in the form $\sum_{i=1}^{r}\binom{x+\lambda_{i}-i}{\lambda_{i}-1}$ for some integer partition $\lambda=(\lambda_1,...,\lambda_r)$ where $v\geq\lambda_{1}\geq...\geq\lambda_{r}\geq1$.
\end{theorem}
The expression above is also known as Gotzmann's decomposition of $p(x)$ \cite{G}.
From this work, other researchers have been able to identify the smoothness of the associated Hilbert schemes as well as the sum of the Betti numbers via $\lambda$ \cite{SS,DLRUZ}.\\
With the rich information embedded in $\lambda$ as well as the advent of software used for experimentation in algebraic geometry such as Macaulay2 \cite{M2} we recognize the utility of having an algorithm for identifying if $p(x)$ is a Hilbert polynomial and if so, recover $\lambda$.\\

\section{Discrete Derivatives}
Let $k$ be a field containing $\QQ$, so that $\NN\subseteq k$.
Then let $S=k^{\NN}$ be the ring of sequences $(f(0),f(1),f(2),...)$ with entries in $k$ where $f:\NN\longrightarrow k$, and addition and subtraction are defined entry-wise.
A polynomial $f\in k[x]$ defines a function $k\rightarrow k$ and by restriction $f|_{\NN}:\NN\longrightarrow k$.
For $f\in S$, we define the $k$-linear operator $\Delta:S\longrightarrow S$ which we refer to as the discrete derivative:
$$(\Delta f)(x)=f(x+1)-f(x)$$
with zero entries filled in as appropriate.
For example, if $f=(18,2,8,2,11,...)$, then 
$$\Delta f=(-16,6,-6,9,...)$$
Before moving on, it is important to note that we refer to a sequence $f\in S$ as a polynomial of degree $d$ if there is a (necessarily unique) polynomial $\hat{f}\in k[x]$ such that $f=\hat{f}|_{\NN}$.
\begin{theorem}
Let $f\in S$ be a sequence.
\begin{enumerate}
    \item $\Delta f=0$ iff $f$ is a constant sequence.
Thus, $\Delta f=\Delta g$ iff $g=f+c$, where $c$ is a constant sequence.
    \item If $f$ is a polynomial of degree $d>0$, then $\Delta f$ is a polynomial of degree $d-1$.
    \item Let $c$ be a constant sequence. For any $n>0$, there is a polynomial $g$ of degree $n$ such that $c=\Delta^{n}g$. 
    \item $\Delta^{n}f=0$ iff $f$ is a polynomial of degree $<n$.
\end{enumerate}
\end{theorem}

\begin{proof}
\hfill\break
\begin{enumerate}
    \item $(\implies)$ Suppose that $\Delta f=0$ and let's assume that $f$ is not a constant sequence.
That is, there exists an $x$ such that $f(x+1)-f(x)\neq0$ which is a contradiction.\\
    $(\impliedby)$ Suppose that $f$ is a constant sequence.
This means that for every $x, f(x+1)-f(x)=0\implies \Delta f=0$.\\
    Furthermore, since $\Delta$ is $k$-linear, $\Delta f=\Delta g$ means that $\Delta f-\Delta g=\Delta(f-g)=0$ iff $f-g=c$ for some constant $c$.
    
    \item Let $f$ be a polynomial of degree $d>0$, say obtained by the restriction of $a_{0}x^{d}+a_{1}x^{d-1}+...$ for $a_{0}\neq0$.
We can the compute the following.
    $$(\Delta f)(x)=f(x+1)-f(x)$$
    $$=a_{0}(x+1)^{d}+a_{1}(x+1)^{d-1}-a_{0}x^{d}-a_{1}x^{d-1}+[\textnormal{lower degree terms}]$$
    $$=a_{0}(x^{d}+dx^{d-1})+a_{1}x^{d-1}-a_{0}x^{d}-a_{1}x^{d-1}+[\textnormal{lower degree terms}]$$
    $$=a_{0}dx^{d-1}+[\textnormal{lower degree terms}]$$
    
    \item Let $g_{0}$ be the sequence described by $g_{0}=x^{n}$.
Be repeated application of $(2)$, $\Delta^{n}g_{0}$ is a constant sequence $a$ (In fact, $a=n!$).
Let $g(x)=\frac{c}{a}g_{0}(x)$ and we can see by $k$-linearity that $\Delta^{n}g=c$.
    
    \item $(\implies)$ Suppose that $\Delta^{n} f=0$ and let's assume that $f$ is a polynomial of degree $d\geq n$.
By repeated application of $(2)$ this implies that $\Delta^{d} f=c$ for some non-zero constant $c$ which is a contradiction.\\
    $(\impliedby)$ Suppose that $f$ is a polynomial of degree $d<n$.
Then by repeated application of $(2)$ we have that $\Delta^{d} f=c$ for some non-zero constant $c$ and by $(1)$ we have that $\Delta^{d+1}f=0$.
Further, since $d+1\leq n$ and $\Delta 0=0$ we have our desired result.
\end{enumerate}
\end{proof}
With these foundational concepts in mind, let's move onto the binomial sequences which we will find have some useful properties.\\
For each $d\geq 0$, define the $d$th binomial sequence $B_{d}$ by
$$B_{d}(x)=\frac{x(x-1)...(x-(d-1))}{d!}=\binom{x}{d}\in k[x]\subseteq S.$$
By convention, $B_{0}$ is the constant sequence equal to 1, and note that $B_{1}(x)=x$.
We may write
$$B_{d}=\left(\underbrace{0,...,0}_{d\textnormal{ times}},\binom{d}{d},\binom{d+1}{d},\binom{d+2}{d},\binom{d+3}{d},...\right).$$
For example,
$$B_{2}=\left(0,0,\binom{2}{2},\binom{3}{2}\binom{4}{2},\binom{5}{2},...\right).$$
Next, so that we can swiftly cite it later on, we will remind the reader of the following well known theorem.
\begin{theorem}
The binomial sequences are integer valued: $B_{d}(\ZZ)\subseteq \ZZ$.
\end{theorem}
In addition to this, we will find that the following theorem will also provide invaluable utility when deriving the algorithm.
\begin{theorem}
$\Delta B_{d}(x)=B_{d-1}(x)$
\end{theorem}
\begin{proof}
This can trivially be shown using Pascals rule.
$$\Delta B_{d}(x)=\binom{x+1}{d}-\binom{x}{d}=\binom{x}{d-1}=B_{d-1}(x)$$
\end{proof}

\section{Discrete Derivative Algorithm}
Now that we have the necessary mechanisms in place from Section 2, let's derive the algorithm.
Moving forward, we will use the falling factorial convention for binomial coefficients.
$$
\binom{n}{k}=
\frac{(n)_{k}}{k!}=
\frac{n(n-1)(n-2)...(n-k+1)}{k!}
$$

Let's repackage the integer partition into the following form (which has been employed in other literature as well \cite{SS}).
$$\lambda=(\lambda_{1}^{r_{1}},\lambda_{2}^{r_{2}},...,\lambda_{e}^{r_{e}})$$
$$=(\underbrace{\lambda_{1},\lambda_{1},...,\lambda_{1}}_{r_{1}\textnormal{ times}},\underbrace{\lambda_{2},\lambda_{2},...,\lambda_{2}}_{r_{2}\textnormal{ times}},...,\underbrace{\lambda_{e},\lambda_{e},...,\lambda_{e}}_{r_{e}\textnormal{ times}})$$
Where $\lambda_{1}>\lambda_{2}>...>\lambda_{e}$ and the size of $\lambda$ is $r=r_{1}+r_{2}+...+r_{e}$.\\
With this new notation, we can now say that if $h(x)$ is a Hilbert polynomial then 
$$h(x)=\sum_{i=1}^{r_{1}}\binom{x+\lambda_{1}-i}{\lambda_{1}-1}+\sum_{i=r_{1}+1}^{r_{1}+r_{2}}\binom{x+\lambda_{2}-i}{\lambda_{2}-1}+...+\sum_{i=r-r_{e}+1}^{r}\binom{x+\lambda_{e}-i}{\lambda_{e}-1}.$$
From here, the general idea of the algorithm is fairly straightforward.
If we consider the sequence $h\in S$ generated by $h(x)$ then by Theorem 4 and $k$-linearity we know that $\Delta^{\lambda_{1}-1} h=r_{1}$.
With this in mind, we count how many times ($\lambda^{*}$) we apply $\Delta$ until we are left with a constant sequence $r$ and subtract the contribution $\sum_{i=1}^{r}\binom{x+\lambda^{*}-i}{\lambda^{*}-1}$ leaving us with a new polynomial sequence in which we apply the same procedure.
In practice however, if the user provides a polynomial $p(x)$ with degree $n$ then we begin by allocating a vector of length $n+1$ with values $\{p(0),p(1),p(2),...,p(n)\}$ and we know that the sequence after an arbitrary number of discrete derivatives will have a unique polynomial associated with it due to the polynomial interpolation theorem and Theorem 2(2).\\
With the above framework in place, let's begin in formalizing the algorithm by defining some subroutines.
First, we define the trivial method $notConstant(p,end)$ which returns $false$ if the sequence $\{p(0),p(1),...,p(end)\}$ is constant (or of size $1$) and $true$ otherwise.
Next, we shall define the $reduce(p)$ which applies $\Delta$ to the polynomial sequence $p$ until we are left with a non-zero constant sequence.\\
\begin{algorithm}[H]
\SetAlgoLined
 {\bf Input}: $p$: Polynomial sequence.\\
 {\bf Input}: $n$: $size(p)-1$.\\
 {\bf Output}: $m$: The minimum number of times $\Delta$ must be applied to $p$ so that $\Delta^{m}p=c$ for some constant $c$.\\
 {\bf Output}: $c$: The value of the constant sequence $\Delta^{m}p=c$.\\
 $counter\gets 0$\;
 $end\gets n$\;
 \While {$notConstant(p, end)$} {
   \For {$i=0$ : $end-1$} {
     $p(x)\gets p(x+1)-p(x)$\;
   }
   $counter\gets counter+1$\;
   $end\gets end-1$\;
 }
 \Return $(counter, p(0))$\;
 \caption{$reduce(p, n)$}
\end{algorithm}
For the sake of addressing some potential concerns it is important to note that if the user passes in a polynomial sequence $p$, a positive integer $n$ and $n$ is the degree of $p$ then the above routine will exit after $notConstant$ checks a sequence of size $1$ due to repeated application of Theorem 2(2).
By the same argument, in cases where $deg(p)<n$ we know that the robustness of the algorithm is also not affected with the only difference being that the $reduce$ method will return after the $notConstant$ routine checks a sequence of size greater than 1.
In fact, we will later find that in the context of the recovery algorithm, these are the only possible use cases.\\
Next, we define the method $subtract(p,n,\lambda,start,end)$ which for $x\in\{0,1,...,n\}$ subtracts $\sum_{i=start}^{end}\binom{x+\lambda-i}{\lambda-1}$ from $p(x)$ and returns the resulting sequence.
Using a telescoping sum with Pascal's rule simplifies this by instead subtracting $\binom{x+\lambda-start+1}{\lambda}-\binom{x+\lambda-end}{\lambda}=\sum_{i=start}^{end}\binom{x+\lambda-i}{\lambda-1}$ from $p(x)$.\\
Finally, let's define the method $isIntegerSeq(p)$ simply as the method which returns $true$ if all the elements in the sequence are in $\ZZ$ and returns $false$ otherwise.
In the final formalization of the algorithm, $isIntegerSeq$ along with one other check will be used to verify that the polynomial provided by the user is in fact a Hilbert polynomial.
With all the subroutines in place, let's define the recovery routine.\\
\begin{algorithm}[H]
\SetAlgoLined
 {\bf Input}: $p(x)\in\QQ[x]$\\
 {\bf Output}: $\lambda$ if $p(x)$ is a Hilbert polynomial and $false$ otherwise.\\
 $n\gets\deg(p(x))$\;
 $p\gets \{p(0),p(1),...,p(n)\}$\;
 $\lambda\gets()$\;
 $s\gets 1$\;
 \If{$isIntegerSeq(p)==false$} {
   \Return $false$\;
 }
 \While {$p\neq\vec{0}$} {
   $p^{*}\gets p$\;
   $(m,r)\gets reduce(p^{*},n)$\;
   \If{$r<0$}{
     \Return $false$\;
   }
   $\lambda\gets(\lambda,(m+1)^{r})$\;
   $p\gets subtract(p,n,m+1,s,s+r-1)$\;
   $s\gets s+r$\;
 }
 \Return $\lambda$\;
 \caption{Discrete Derivative Recovery Algorithm}
\end{algorithm}
Now that the algorithm has been defined, let's address some potential questions regarding the algorithm.
First, let's address the first ``if" statement in the algorithm.
This check addresses two concerns.
The first being that by Theorem 3 it must be the case that if $p(x)$ is a Hilbert polynomial then $\forall x\in\ZZ,p(x)\in\ZZ$.
Furthermore, simply checking the values $\{p(0),p(1),...,p(\deg(p(x)))\}$ is sufficient enough due to the following theorem.\\
\begin{theorem}
\cite{CC}
A degree $n$ polynomial with rational coefficients is integer-valued iff it takes integer values on $n+1$ consecutive integer values.
\end{theorem}
In addition to the previously mentioned concern, the first ``if" statements also addresses ideas surrounding the following theorem.
\begin{theorem}
Let $(m,r)=reduce(p,n)$ where $p$ is an polynomial sequence with values in $\QQ$.
If $r\in \QQ\backslash \ZZ$ then at least one of the values in $\{p(0),...,p(n)\}$ must be in $\QQ\backslash \ZZ$.
\end{theorem}
\begin{proof}
Suppose that $r\in\QQ\backslash \ZZ$ and let's assume that all the values in $\{p(0),...,p(n)\}$ are in $\ZZ$ then since $\ZZ$ is a subring of $\QQ$, and $reduce$ generates $r$ by simply subtracting elements in $p$ from other elements in $p$ then $r\in\ZZ$ which is a contradiction.
\end{proof}
The importance of this theorem is realized when we note that the leading coefficient of a Hilbert polynomial $h(x)$ is $\frac{r_{1}}{(\lambda_{1}-1)!}$ for some positive integer $r_{1}$ and by the observations in the proof of Theorem 2(3) we know that the second value in the tuple returned by the $reduce(h,n)$ method is $r_{1}$.
It is also important to note that checking if $p$ is all integers is only needed in the beginning of the algorithm since if that is the case then $p$ will remain integer valued due to the $subtract$ method only subtracting integer contributions from $p$.\\
Moreover, the check for $r<0$ in the ``if" statement inside the ``while" loop is needed since as we mentioned earlier, the leading coefficient of a Hilbert polynomial must be a \textit{positive} integer multiple of $\frac{1}{(\lambda_{1}-1)!}$.\\
Now let's observe the asymptotic runtime of this algorithm assuming that we are in the unit cost model and $n=deg(p(x))$.
The outer while loop runs $\mathcal{O}(n)$ times.
The $reduce$ step runs in $\mathcal{O}(n^{2})$ time.
The $subtract$ step runs in $\mathcal{O}(n^{2})$.
In total, the algorithm runs in $\mathcal{O}(n^{3})$

\section{Improved Algorithm}

In the following section, we describe a more efficient algorithm to identify and recover $\lambda$ from a Hilbert polynomial $h(x)\in\mathbb{Q}[x]$ by simply representing the polynomial in a different form.
That is, we will rewrite a degree $n$ polynomial $h(x)$ in the following form which comes from Newton's forward difference formula.

$$
h(x)=\sum_{i=0}^{n}\binom{x}{i}\Delta^{i}h(0)
$$

Notice that for a degree $n$ Hilbert polynomial, the number of occurences of $n+1$ in $\lambda$ is $\Delta^{n}h(0)$.
If we wish to subtract the expression $\sum_{i=a}^{b}\binom{x+\lambda-i}{\lambda-1}$ then we rewrite it as the following by using a telescoping sum with Pascal's rule and Vandermonde's identity.

$$
\sum_{i=a}^{b}\binom{x+\lambda-i}{\lambda-1}=
\binom{x+\lambda-a+1}{\lambda}-\binom{x+\lambda-b}{\lambda}
$$

$$
=\sum_{i=0}^{\lambda}\binom{\lambda-a+1}{\lambda-i}\binom{x}{i}-
\sum_{i=0}^{\lambda}\binom{\lambda-b}{\lambda-i}\binom{x}{i}
$$

We now have a form in which we can easily subtract from the proposed form of our polynomial.
Let's begin by setting up subroutines and conventions for our improved algorithm.
First, we will represent our input (a degree $n$ polynomial $h(x)$) as a sequence of length $n+1$ where the $i$th entry is $\Delta^{i}h(0)$.
We will denote such sequence as $\eta$ where the $i$th entry is $\eta(i)$.
Next, we will define the helper function $subtractNewton(\eta,\lambda,a,b)$ which will subtract the contribution $\sum_{i=a}^{b}\binom{x+\lambda-i}{\lambda-1}$ from our polynomial in the form represented by $\eta$.

\begin{algorithm}[H]
\SetAlgoLined
 {\bf Input}: $\eta$: The coefficient vector of a polynomial represented by Newton's forward difference formula.\\
 {\bf Input}: $\lambda$: The part of the partition whose contribution is to be removed.\\
 {\bf Input}: $a,b$: The first and last indices of the sum, where $a\leq b$.\\
 {\bf Output}: The coefficient vector of the polynomial sequence with $\sum_{i=a}^{b}\binom{x+\lambda-i}{\lambda-1}$ subtracted.\\
 \For {$i=0$ : $\lambda-1$} {
   $\eta(i)\gets \eta(i)-\left[\binom{\lambda-a+1}{\lambda-i}-\binom{\lambda-b}{\lambda-i}\right]$\;
 }
 \Return $\eta$\;
 \caption{$subtractNewton(\eta, \lambda, a, b)$}
\end{algorithm}

As a matter of fact, $subtractNewton$ can be optimized further into the following routine which is equivalent in functionality.

\begin{algorithm}[H]
\SetAlgoLined
 {\bf Input}: $\eta$: The coefficient vector of a polynomial represented by Newton's forward difference formula.\\
 {\bf Input}: $\lambda$: The part of the partition whose contribution is to be removed.\\
 {\bf Input}: $a,b$: The first and last indices of the sum, where $a\leq b$.\\
 {\bf Output}: The coefficient vector of the polynomial sequence with $\sum_{i=a}^{b}\binom{x+\lambda-i}{\lambda-1}$ subtracted.\\
 $num_{a}\gets \lambda-a+1$\;
 $num_{b}\gets \lambda-b$\;
 $div\gets 1$\;
 \For {$i=0$ : $\lambda-1$} {
   $\eta(\lambda-1-i)\gets \eta(\lambda-1-i)-\left[\frac{num_{a}}{div}-\frac{num_{b}}{div}\right]$\;
   $num_{a}\gets num_{a}\cdot (\lambda-a-i)$\;
   $num_{b}\gets num_{b}\cdot (\lambda-b-i-1)$\;
   $div\gets div\cdot(i+2)$
 }
 \Return $\eta$\;
 \caption{$subtractNewton(\eta, \lambda, a, b)$}
\end{algorithm}

We will also define a trivial function $topDegree(\eta,n)$ which returns the largest index $0\leq i\leq n$ s.t. $\eta(i)\neq 0$ and returns -1 if such index doesn't exist.
All in all, the improved algorithm becomes the following.

\begin{algorithm}[H]
\SetAlgoLined
 {\bf Input}: $p(x)\in\QQ[x]$\\
 {\bf Output}: $\lambda$ if $p(x)$ is a Hilbert polynomial and $false$ otherwise.\\
 $n\gets\deg(p(x))$\;
 $p\gets (p(0),p(1),...,p(n))$\;
 \If{$isIntegerSeq(p)==false$} {
   \Return $false$\;
 }
 $\eta\gets (\Delta^{0}p(0),\Delta^{1}p(0),...,\Delta^{n}p(0))$\;
 $\lambda\gets()$\;
 $s\gets 1$\;
 $m\gets topDegree(\eta,n)$\;
 \While {$m\geq 0$} {
   $r\gets \eta(m)$\;
   \If{$r<0$}{
     \Return $false$\;
   }
   $\lambda\gets(\lambda,(m+1)^{r})$\;
   $\eta\gets subtractNewton(\eta,m+1,s,s+r-1)$\;
   $s\gets s+r$\;
   $m\gets topDegree(\eta,m-1)$\;
 }
 \Return $\lambda$\;
 \caption{Newton Basis Recovery Algorithm}
\end{algorithm}

Once again, let's observe the asymptotic runtime of this algorithm assuming that we are in the unit cost model, $n=deg(p(x))$, and using Algorithm 4 for the $subtractNewton$ method.
The outer while loop executes $\mathcal{O}(n)$ times.
The $subtractNewton$ and $topDegree$ routines have a time complexity of $\mathcal{O}(n)$.
In total, the runtime is $\mathcal{O}(n^{2})$.

\section{Conclusions}
In this paper, we reviewed the topic of polynomial sequences, discrete derivatives, and some useful properties of binomial sequences.
By utilizing discrete derivatives, we then devised a robust algorithm that classifies Hilbert polynomials.
Finally, by utilizing a different form to write our input polynomial, we derived an improved algorithm.

\bibliographystyle{plain}
\bibliography{bibliography}
\end{document}